\documentclass{llncs}
\usepackage{amsmath}
\usepackage{amssymb}
\usepackage{enumerate}
\usepackage{varioref}
\usepackage{charter,eulervm}
\usepackage{graphicx}
\usepackage{subfig,wrapfig}
\usepackage{boxedminipage}

\title{Results on independent sets in categorical products of 
graphs,  
the ultimate categorical independence ratio and the 
ultimate categorical independent domination ratio}

\author{
 Wing-Kai~Hon\inst{1}
\and
 Ton~Kloks
\and 
 Hsiang-Hsuan~Liu\inst{1}
\and 
 Sheung-Hung~Poon\inst{1} 
\and
 Yue-Li~Wang\inst{2}} 
\institute{
 Department of Computer Science\\
 National Tsing Hua University, Taiwan\\
 {\tt \{wkhon,hhliu,spoon\}@cs.nthu.edu.tw} 
\and 
 Department of Information Management\\ 
 National Taiwan University of Science and Technology\\ 
 {\tt ylwang@cs.ntust.edu.tw}
}

\begin{document}

\maketitle

\begin{abstract}
We show that there are polynomial-time 
algorithms to compute maximum independent sets  
in the categorical products of two cographs and 
two splitgraphs. The ultimate categorical 
independence ratio of a graph $G$ is defined as 
$\lim_{k \rightarrow \infty} \alpha(G^k)/n^k$. 
The ultimate categorical 
independence ratio is polynomial  
for cographs, permutation graphs, interval graphs, graphs 
of bounded treewidth and splitgraphs.  
When $G$ is a planar graph of maximal degree three then  
$\alpha(G \times K_4)$ is NP-complete. We present a PTAS 
for the ultimate categorical independence ratio of 
planar graphs.  
We present an $O^{\ast}(n^{n/3})$ exact, exponential algorithm 
for general graphs. 
We prove that 
the ultimate categorical independent domination 
ratio for complete multipartite graphs is zero, 
except when the graph is complete bipartite with  
color classes of equal size 
(in which case it is $1/2$). 
\end{abstract}

\section{Introduction}
\label{section intro}

Let $G$ and $H$ be two graphs. The categorical product also 
travels under the guise of tensor product, or direct product, 
or Kronecker product, and even more names have been given to it. 
It is defined as follows. It is a graph, denoted as 
$G \times H$. Its vertices are the ordered pairs 
$(g,h)$ where $g \in V(G)$ and $h \in V(H)$. Two of its 
vertices, say $(g_1,h_1)$ and $(g_2,h_2)$ are adjacent 
if 
\[\{\;g_1,\;g_2\;\} \in E(G) 
\quad\text{and}\quad \{\;h_1,\;h_2\;\} \in E(H).\]  

\bigskip 

One of the reasons for its popularity is Hedetniemi's 
conjecture, which is now more 
than 40 years old~\cite{kn:hedetniemi,kn:sauer,kn:tardif,kn:zhu}. 
\begin{conjecture}
For any two graphs $G$ and $H$ 
\[\boxed{\chi(G \times H) = \min\;\{\;\chi(G),\;\chi(H)\;\}.}\] 
\end{conjecture}
It is easy to see that the right-hand side is an upperbound. 
Namely, if $f$ is a vertex coloring of $G$ then one can color 
$G \times H$ by defining a coloring $f^{\prime}$ 
as follows 
\[f^{\prime}((g,h))=f(g), 
\quad\text{for all $g \in V(G)$ and $h \in V(H)$.}\]  
Recently, it was shown that the fractional version of 
Hedetniemi's conjecture is true~\cite{kn:zhu2}. 

\bigskip 

When $G$ and $H$ are perfect then Hedetniemi's conjecture 
is true. Namely, let $K$ be a clique of cardinality 
at most 
\[|K| \leq \min \;\{\;\omega(G),\;\omega(H)\;\}.\] 
It is easy to check that $G \times H$ has a clique 
of cardinality $|K|$. One obtains an `elegant' proof 
via homomorphisms as follows.  
By assumption, there exist homomorphisms $K \rightarrow G$ and 
$K \rightarrow H$. 
This implies that there also is a 
homomorphism $K \rightarrow G \times H$ (see, eg,~\cite{kn:hahn,kn:hell}). 
(Actually, if $W$, $P$ and $Q$ are any graphs, 
then there exist homomorphisms $W \rightarrow P$ and 
$W \rightarrow Q$ if and only if there exists a 
homomorphism $W \rightarrow P \times Q$.)
In other words~\cite[Observation~5.1]{kn:hahn}, 
\[\omega(G \times H) \geq \min\;\{\;\omega(G),\;\omega(H)\;\}.\] 
Since $G$ and $H$ are perfect, $\omega(G)=\chi(G)$ and 
$\omega(H)=\chi(H)$. This proves the claim, since 
\begin{equation}
\begin{split}
\chi(G \times H) \geq \omega(G \times H) \geq 
& \min\;\{\;\omega(G),\;\omega(H)\;\} \\
= & \min\;\{\;\chi(G),\;\chi(H)\;\} 
\geq \chi(G \times H).
\end{split}
\end{equation}

\bigskip 

Much less is known about the independence 
number of $G \times H$. It is easy to see that 
\begin{equation}
\label{eqn0}
\alpha(G \times H) \geq \max\;\{\;\alpha(G) \cdot |V(H)|,\;
\alpha(H) \cdot |V(G)|\;\}.
\end{equation}
But this lowerbound can be arbitrarily bad, even for 
threshold graphs~\cite{kn:jha}. For any graph $G$ and any 
natural number $k$ there exists a threshold graph $H$ 
such that 
\[\alpha(G \times H) \geq k+L(G,H),\] 
where 
$L(G,H)$ is the lowerbound expressed in~\eqref{eqn0}. 
Zhang recently proved that, 
when $G$ and $H$ are vertex transitive then equality holds 
in~\eqref{eqn0}~\cite{kn:zhang}. Notice that, 
when $G$ is vertex transitive then $G^k$\footnote{Here we write 
$G^k$ for the $k$-fold product $G \times \dots \times G$.} is also 
vertex transitive and so, 
by the ``no-homomorphism'' lemma of Albertson 
and Collins,  
$\alpha(G^k)=\alpha(G) \cdot n^{k-1}$.  

\bigskip 

\begin{definition}
A graph is a cograph if it has no induced 
$P_4$, ie, a path with four vertices. 
\end{definition}
Cographs are characterized by the property that 
every induced subgraph $H$ satisfies one of 
\begin{enumerate}[\rm (a)]
\item $H$ has only one vertex, or  
\item $H$ is disconnected, or 
\item $\Bar{H}$ is disconnected. 
\end{enumerate}
It follows that cographs can be represented 
by a cotree. This is pair $(T,f)$ where $T$ is a 
rooted tree and $f$ is a 1-1 map from the vertices 
of $G$ to the leaves of $T$. Each internal node of $T$, 
including the root,  
is labeled as $\otimes$ or $\oplus$. When the label 
is $\oplus$ then the subgraph $H$, induced by the vertices 
in the leaves, is disconnected. Each child of the 
node represents one component. When the node is labeled 
as $\otimes$ then the complement of the induced 
subgraph $H$ is disconnected. In that case, each component 
of the complement is represented by one child of the node. 

When $G$ is a cograph then a cotree for $G$ can be obtained 
in linear time. 

\bigskip 

Cographs are perfect, see, eg,~\cite[Section~3.3]{kn:kloks2}. 
When $G$ and $H$ are cographs then $G \times H$ is not 
necessarily perfect. For example, when $G$ is the paw, 
ie, $G \simeq K_1 \otimes (K_2 \oplus K_1)$ then $G \times K_3$ 
contains an induced $C_5$~\cite{kn:ravindra}. Ravindra and 
Parthasarathy characterize the pairs $G$ and $H$ for which 
$G \times H$ is perfect~\cite[Theorem~3.2]{kn:ravindra}. 

\section{Independence in categorical products of cographs}

It is well-known that $G \times H$ is connected if and only 
if both $G$ and $H$ are connected and at least one of them is not 
bipartite~\cite{kn:weichsel}. When $G$ and $H$ are connected 
and bipartite, then $G \times H$ consists of two components. 
In that case, two vertices $(g_1,h_1)$ and $(g_2,h_2)$ 
belong to the same component if the distances $d_G(g_1,g_2)$ 
and $d_H(h_1,h_2)$ have the same parity.  

\bigskip 

\begin{definition}
The rook's graph $R(m,n)$ is the linegraph of 
the complete bipartite graph $K_{m,n}$. 
\end{definition}
The rook's graph $R(m,n)$ has as its vertices 
the vertices of the grid, 
$(i,j)$, with $1 \leq i \leq m$ and $1 \leq j \leq n$. 
Two vertices are adjacent if they are in the same 
row or column of the grid. 
The rook's graph is perfect, since all linegraphs 
of bipartite graphs are perfect (see, eg,~\cite{kn:kloks2}). 
By the perfect graph theorem, also the complement 
of rook's graph is perfect. 

\begin{proposition}
\label{lm rook}
Let $m,n \in \mathbb{N}$. Then 
\[K_m \times K_n \simeq \Bar{R},\]
where $\Bar{R}$ is the complement of the 
rook's graph $R=R(m,n)$.
\end{proposition}

\begin{lemma}
\label{lm cmp}
Let $G$ and $H$ be complete multipartite. 
Then $G \times H$ is perfect. 
\end{lemma}
\begin{proof}
Ravindra and Parthasarathy prove that $G \times H$ is 
perfect if and only if either 
\begin{enumerate}[\rm (a)]
\item $G$ or $H$ is bipartite, or 
\item Neither $G$ nor $H$ contains an induced odd 
cycle of length at least 5 nor an induced paw. 
\end{enumerate}
Since $G$ and $H$ are perfect, they do not contain 
an odd hole. 
Furthermore, the complement of $G$ and $H$ is a 
union of cliques, and so the complements are $P_3$-free. 
The complement of a paw is $K_1 \oplus P_3$ and so 
it has an induced $P_3$. This proves the claim. 
\qed\end{proof}

\bigskip 

Let $G$ and $H$ be complete multipartite. Let $G$ be the 
join of $m$ independent sets, say with $p_1, \dots, p_m$ vertices, 
and let $H$ be the join of 
$n$ independent sets, say with $q_1,\dots,q_n$ vertices.  
We shortly describe how $G \times H$ is obtained from 
the complement of the rook's graph $R(m,n)$. We call the 
structure a generalized rook's graph.  

\bigskip 

Each vertex $(i,j)$ in $R(m,n)$ is replaced by an 
independent set $I(i,j)$   
of cardinality $p_i \cdot q_j$. Denote the vertices 
of this independent set as 
\[(i_s,j_t) \quad\text{where $1 \leq s \leq p_i$ and 
$1 \leq t \leq q_j$.}\] 
Two vertices $(i_s,j_t)$ and $(i^{\prime}_s,j_t)$ 
are adjacent and these types of row- and column-adjacencies 
are the only adjacencies in this generalized rook's graph. 
The graph $G \times H$ is obtained from the partial complement 
of the generalized rook's graph. 

\bigskip 

\begin{lemma}
Let $G$ and $H$ be complete multipartite graphs. 
Then 
\begin{equation}
\label{eqn1}
\alpha(G \times H) = \kappa(G \times H)=
\max\;\{\;\alpha(G) \cdot |V(H)|,\;\alpha(H) \cdot |V(G)|\;\}.
\end{equation}
\end{lemma}
\begin{proof}
Two vertices $(g_1,h_1)$ and $(g_2,h_2)$ are 
adjacent if $g_1$ and $g_2$ are not in a common 
independent set in $G$ and $h_1$ and $h_2$ are 
not in a common independent set in $H$. 

\medskip 

\noindent
Let $\Omega$ be a maximum independent set of $G$. 
Then 
\[\{\;(g,h)\;|\; g \in \Omega \quad\text{and}\quad h \in V(H)\;\}\] 
is an independent set in $G \times H$. 
We show that all maximal independent sets are of this form 
or of the symmetric form with $G$ and $H$ interchanged. 

\medskip 

\noindent
Consider the complement of the 
rook's graph. Any independent set 
must have all its vertices in one row or in 
one column. 
This shows that every maximal independent set in $G \times H$ 
is a generalized row or column in the rook's graph. 
Since the graphs are perfect, the number of cliques 
in a clique cover of $G \times H$ equals $\alpha(G \times H)$. 
\qed\end{proof}

\bigskip 

\begin{remark}
Notice that complete multipartite graphs are  
not vertex transitive, unless all independent sets 
have the same cardinality. 
\end{remark}

\begin{proposition}
\label{lm union}
Let $G$ and $H$ be cographs and assume that 
$G$ is disconnected. Say that $G=G_1 \oplus G_2$. 
Then 
\[\alpha(G \times H)=\alpha(G_1 \times H)+\alpha(G_2 \times H).\]
\end{proposition}

\begin{lemma}
\label{lm join}
Let $G$ and $H$ be connected cographs. 
Say $G=G_1 \otimes G_2$ and $H=H_1 \otimes H_2$. 
Then 
\[\alpha(G \times H)=\min \;\{\;\alpha(G_1 \times H),\;
\alpha(G_2 \times H),\;\alpha(G \times H_1),\;\alpha(G \times H_2)\;\}.\] 
\end{lemma}
\begin{proof}
Every vertex of $V(G_1) \times V(H_1)$ is adjacent to 
every vertex of $V(G_2) \times V(H_2)$ and, likewise, 
every vertex of $V(G_1) \times V(H_2)$ is adjacent to 
every vertex of $V(G_2) \times V(H_1)$. 
This proves the claim. 
\qed\end{proof}

\begin{theorem}
\label{thm cographs}
There exists an $O(n^2)$ algorithm 
which computes $\alpha(G \times H)$ when $G$ and $H$ 
are cographs. 
\end{theorem}
\begin{proof}
The proof follows easily from 
Proposition~\ref{lm union} and Lemma~\ref{lm join}.
\qed\end{proof}

\begin{remark}
It seems not easy to extend the result of Theorem~\ref{thm cographs} 
to higher dimensions. 
It would be interesting to know whether $\alpha(G^k)$, 
for $k \in \mathbb{N}$,  
is computable in polynomial time when $G$ is a cograph. 
Even for $k=3$ we have no answer.
\end{remark}
 
\section{Splitgraphs}

F\"oldes and Hammer introduced splitgraphs~\cite{kn:foldes}. 
We refer to~\cite[Chapter~6]{kn:golumbic} and~\cite{kn:merris} for 
some background information on this class of graphs. 

\begin{definition}
A graph $G$ is a splitgraph if there is a partition 
$\{S,C\}$ of its vertices such that $G[C]$ is a clique 
and $G[S]$ is an independent set. 
\end{definition}

\begin{theorem}
Let $G$ and $H$ be splitgraphs. There exists a 
polynomial-time algorithm to compute the independence 
number of $G \times H$. 
\end{theorem}
\begin{proof}
Let $\{S_1,C_1\}$ and $\{S_2,C_2\}$ be the 
partition of $V(G)$ and $V(H)$, respectively, 
into independent sets and cliques. Let 
$c_i=|C_i|$ and $s_i=|S_i|$ for $i \in \{1,2\}$.  
The vertices of $C_1 \times C_2$ form a rook's graph. 

\medskip 

\noindent
We consider three cases. First consider  
the maximum independent 
sets without any vertex of $V(C_1) \times V(C_2)$. 
Notice that the subgraph of $G \times H$ induced 
by the vertices of 
\[V(S_1) \times V(C_2) \cup V(C_1) \times V(S_2) \cup V(S_1) \times V(S_2)\] 
is bipartite. A maximum independent set in a bipartite graph can be 
computed in polynomial time. 

\medskip 

\noindent
Consider maximum independent sets that contain exactly one 
vertex $(c_1,c_2)$ of $V(C_1) \times V(C_2)$. The maximum 
independent set of this type can be computed as follows. 
Consider the bipartite graph of 
the previous case and remove the neighbors of $(c_1,c_2)$ 
from this graph. 
The remaining graph is bipartite. Maximizing over all 
pairs $(c_1,c_2)$ gives the maximum independent set of this 
type. 

\medskip 

\noindent
Consider maximum independent sets that contain at least 
two vertices of the rook's graph $V(C_1) \times V(C_2)$. 
Then the two vertices must be in one row or in one 
column of the grid, since otherwise they are adjacent. 
Let the vertices of the independent set be contained 
in row $c_1 \in V(C_1)$. Then the vertices of 
$V(S_1) \times V(C_2)$ of the independent set 
are contained in 
\[W=\{\;(s_1,c_2)\;|\; s_1 \notin N_G(c_1) \quad\text{and}\quad 
c_2 \in C_2\;\}.\] 
Consider the bipartite graph with one color class 
defined as the following set of vertices 
\[\{\;(c_i,s_2)\;|\; c_i \in C_1 \;\text{and}\; s_2 \in S_2\;\}
\cup \{\;(s_1,s_2)\;|\; 
s_1 \in V(S_1) \;\text{and}\; s_2 \in V(S_2)\;\},\]   
and the other color class defined as  
\[W \cup \{\;(c_1,c_2)\;|\; c_2 \in C_2\;\}.\]  
Since this graph is bipartite, the maximum independent set 
of this type can be computed in polynomial time by maximizing 
over the rows $c_1 \in C_1$ and columns $c_2 \in C_2$. 

\medskip 

\noindent
This proves the theorem. 
\qed\end{proof}

\section{Tensor capacity}

In this section we consider the powers of a graph 
under the categorical product. 

\bigskip 

\begin{definition}
The independence ratio of a graph $G$ is 
defined as 
\begin{equation}
\label{eqn2}
r(G) = \frac{\alpha(G)}{|V(G)|}.
\end{equation}
\end{definition}
For background information on the related Hall-ratio we 
refer to~\cite{kn:simonyi,kn:toth5}. 

\bigskip 

By~\eqref{eqn0} for any two graphs $G$ and $H$ we 
have 
\begin{equation}
\label{eqn3}
r(G \times H) \geq \max\;\{\;r(G),\;r(H)\;\}.
\end{equation}
It follows that $r(G^k)$ is non-decreasing. Also, it is bounded 
from above by 1 and so the limit when $k \rightarrow \infty$ exists. 
This limit was introduced in~\cite{kn:brown} as the 
`ultimate categorical independence ratio.' See 
also~\cite{kn:alon,kn:hahn2,kn:hell2,kn:lubetzky}. 
For simplicity we call it the 
tensor capacity of a graph. 
Alon and Lubetzky, and also T\'oth claim that computing the 
tensor capacity is NP-complete but, unfortunately neither 
provides a proof~\cite{kn:alon,kn:lubetzky,kn:toth2}. 

\begin{definition}
Let $G$ be a graph. The \underline{tensor capacity} of $G$ is 
\begin{equation}
\label{eqn4}
\boxed{\Theta^T(G) = \lim_{k \rightarrow \infty} r(G^k).} 
\end{equation}
\end{definition}

\bigskip 

Hahn, Hell and Poljak prove that for the {\em Cartesian product\/},  
\[\frac{1}{\chi(G)} \leq \lim_{k \rightarrow \infty} r(\Box^k G) 
\leq \frac{1}{\chi_f(G)},\] 
where $\chi_f(G)$ is the fractional chromatic number of $G$~\cite{kn:hahn2}. 
This shows that it is computable in polynomial time 
for graphs that satisfy $\omega(G)=\chi(G)$. 

\bigskip 

Brown et al.~\cite[Theorem 3.3]{kn:brown} 
obtain the following lowerbound for the 
tensor capacity. 
\begin{equation}
\label{eqn5}
\Theta^T(G) \geq a(G) 
\quad\text{where}\quad a(G)=\max_{\text{$I$ an independent set}}\; 
\frac{|I|}{|I| + |N(I)|}. 
\end{equation}
It is related to the binding 
number $b(G)$ of the graph $G$. Actually, the binding number 
is less than 1 if and only if $a(G) > \frac{1}{2}$. 
In that case, the binding number is realized by an independent set 
and it is equal to $b(G)=\frac{1-a(G)}{a(G)}$~\cite{kn:kloks,kn:toth2}. 
The binding number is computable in 
polynomial time~\cite{kn:cunningham,kn:kloks,kn:woodall}. 
See also Corollary~\ref{cor pol a^T} below.    

\bigskip 

The following proposition was proved in~\cite{kn:brown}. 

\begin{proposition}
If $r(G) > \frac{1}{2}$ then $\Theta^T(G)=1$. 
\end{proposition}
Therefore, a better lowerbound for $\Theta^T(G)$ is 
provided by 
\begin{equation}
\label{eqn6}
\Theta^T(G) \geq a^{\ast}(G)=\begin{cases}
a(G) & \quad\text{if $a(G) \leq \frac{1}{2}$}\\
1 & \quad\text{if $a(G) > \frac{1}{2}$.}
\end{cases}
\end{equation}

\bigskip 

\begin{definition}
Let $G=(V,E)$ be a graph. A fractional matching 
is a function $f: E \rightarrow \mathbb{R}^{+}$, which assigns 
a non-negative real number to each edge, such that 
for every vertex $x$ 
\[\sum_{e \ni x} \; f(e) \leq 1.\]
A fractional matching $f$ is perfect if it achieves the 
maximum 
\[f(E)=\sum_{e \in E}\; f(e) = \frac{|V|}{2}.\] 
\end{definition}
 
Alon and Lubetzky proved the following theorem in~\cite{kn:alon} 
(see also~\cite{kn:kloks}). 

\begin{theorem}
For every graph $G$ 
\begin{equation}
\label{eqn7}
\Theta^T(G)=1 \;\Leftrightarrow\; a^{\ast}(G)=1 \;
\Leftrightarrow\; \text{$G$ has no fractional perfect matching.}
\end{equation}
\end{theorem}

\begin{corollary}
\label{cor pol a^T}
There exists a polynomial-time algorithm to decide 
whether 
\[\Theta^T(G) =1 \quad\text{or}\quad \Theta^T(G) \leq \frac{1}{2}.\]  
\end{corollary}

\bigskip 
 
The following theorem was raised as a 
question by Alon and Lubetzky  
in~\cite{kn:alon,kn:lubetzky}. The theorem was proved by 
\'Agnes T\'oth~\cite{kn:toth2}. 

\begin{theorem}
\label{thm toth}
For every graph $G$ 
\[ \Theta^T(G)=a^{\ast}(G).\]
Equivalently, every graph $G$ satisfies  
\begin{equation}
\label{eqn8}
a^{\ast}(G^2)=a^{\ast}(G). 
\end{equation}
\end{theorem}
T\'oth proves that 
\begin{equation}
\label{eqn11}
\text{if $a(G) \leq \frac{1}{2}$ or $a(H) \leq \frac{1}{2}$ then}
\quad 
a(G \times H) \leq \max\;\{\;a(G),\;a(H)\;\}.
\end{equation}
Actually, T\'oth shows that, if $I$ is an independent 
set in $G \times H$ then   
\[|N_{G \times H}(I)| \geq |I| \cdot \min\;\{\;b(G),\;b(H)\;\}.\]
{F}rom this, Theorem~\ref{thm toth} easily follows. 
As a corollary 
(see~\cite{kn:alon,kn:lubetzky,kn:toth2})   
one obtains that, 
for any two graphs $G$ and $H$
\[r(G \times H) \leq \max\;\{\;a^{\ast}(G),\;a^{\ast}(H)\;\}.\] 

\bigskip 

T\'oth also proves the 
following theorem in~\cite{kn:toth2}. 
This was conjectured by 
Brown et al.~\cite{kn:brown}. 

\begin{theorem} 
\label{thm tensor cap union}
For any two graphs $G$ and $H$, 
\begin{equation}
\label{eqn9}
\Theta^T(G \oplus H)=\max\;\{\;\Theta^T(G),\;\Theta^T(H)\;\}.
\end{equation}
\end{theorem}
Notice that the analogue of this statement, 
with $a^{\ast}$ 
instead of $\Theta^T$, is straightforward. 
The first part of the following theorem 
was proved by 
Alon and Lubetzky in~\cite{kn:alon}. 

\begin{theorem}
\label{lm tensor cap join}
 For any two graphs $G$ and $H$, 
\begin{equation}
\label{eqn10}
\Theta^T(G \oplus H) = \Theta^T(G \times H)= 
\max\;\{\;\Theta^T(G),\;\Theta^T(H)\;\}.
\end{equation}
\end{theorem}

\bigskip 

For cographs we obtain the following theorem. 

\begin{theorem}
There exists an efficient algorithm to compute 
the tensor capacity for cographs. 
\end{theorem}
\begin{proof}
By Theorem~\ref{thm toth} it is sufficient to 
compute $a(G)$, as defined in~\eqref{eqn5}. 

\medskip 

\noindent 
Consider a cotree for $G$. For each node the algorithm 
computes a table. The table contains numbers  
$\ell(k)$, for $k \in \mathbb{N}$, where 
\[\ell(k)= min\;\{\;|N(I)|\;|\; \text{$I$ is an independent set with 
$|I|=k$}\;\}.\] 
Notice that $a(G)$ can be obtained from 
the table at the root node via 
\[a(G)=\max_k \; \frac{k}{k+\ell(k)}.\] 

\medskip 

\noindent
Assume $G$ is the union of two cographs $G_1 \oplus G_2$. 
An independent set $I$ is the union of two 
independent sets $I_1$ in $G_1$ and $I_2$ in $G_2$. 
Let the table entries for $G_1$ and $G_2$ be 
denoted by the functions $\ell_1$ and $\ell_2$.  
Then 
\[\ell(k)=\min\;\{\;\ell_1(k_1) + \ell_2(k_2)\;|\; k_1+k_2=k\;\}.\] 

\medskip 

\noindent
Assume that $G$ is the join of two cographs, say 
$G = G_1 \otimes G_2$. 
An independent set in $G$ can have vertices in at most 
one of $G_1$ and $G_2$. Therefore, 
\[\ell(k) = \min\;\{\;\ell_1(k)+|V(G_2)|,\; 
\ell_2(k)+|V(G_1)|\;\}.\] 

\medskip 

\noindent
This proves the theorem. 
\qed\end{proof}

\begin{remark}
The tensor capacity is computable in 
polynomial time for many other classes of graphs 
via similar methods. 
We describe algorithms for some classes of graphs 
in Appendix~\ref{section classes}.  
\end{remark}

\section{An exact exponential algorithm for the tensor capacity}

Let $G$ be a splitgraph with a partition $\{S,C\}$ of 
its $n$ vertices such that $G[C]$ is a clique 
and $G[S]$ is an independent set.  
For any independent set $I$ of $G$, $I$ can contain 
at most one vertex from $C$.  Define, for $i \in\{0,1\}$,
\[a_i(G)=\max\;\left\{\;
\;\frac{|I|}{|I|+|N(I)|}\;\mid \; 
\text{$I$ an independent set with $|C \cap I|=i$}\;\right\}\] 
Then the value $a(G)$
is obtained by 
\[\max\; \{\; a_0(G), a_1(G)\; \}.\] 

\bigskip 

To compute $a_0(G)$, we shall make use 
of the following simple observation:  
If $S$ can be partitioned into two sets $S_1$ and $S_2$, 
such that their neighbor sets $N(S_1)$ and $N(S_2)$ are disjoint, 
then there exists an optimal $I^*$ for $a_0(G)$, such 
that $I^* \subseteq S_1$ or $I^* \subseteq S_2$.  
To see this, suppose that it is not the case.  
Then, by assumption we can partition $I^*$ into non-empty 
sets $I_1 = I^* \cap S_1$ and $I_2 = I^* \cap S_2$, and we have 
$|I^*| = |I_1| + |I_2|$ and $|N(I^*)| = |N(I_1)| + |N(I_2)|$. 
Then 
\[ a_0(G) = \frac{|I^*|}{|I^*|+|N(I^*)|} \leq 
\max\;\left \{\; \frac{|I_1|}{|I_1|+|N(I_1)|}, 
\frac{|I_2|}{|I_1|+|N(I_2)|} \;\right\} \leq a_0(G).\]
This proves the claim. 

\bigskip 

Based on this observation, we modify a 
technique described by, eg, Cunningham~\cite{kn:cunningham}, 
that transforms the problem into a max-flow (min-cut) problem.  
We construct a flow network $F$ with vertices corresponding 
to each vertex of $S$ and $C$, a source vertex $s$ and 
a sink vertex $t$.  
We make the source $s$ adjacent to each 
vertex in $S$, with capacity 1, 
and the sink $t$ adjacent to each vertex in $C$, 
with capacity 1 as well.  
In addition, if $u \in S$ and $v \in C$ 
are adjacent in the original graph $G$, 
the corresponding vertices are adjacent in $F$, 
with capacity set to $\infty$.  
Note that we omit the edges between vertices 
in $C$. 

\bigskip 

Consider a minimum $s$-$t$ cut in $F$.  
Let $S_1$ be the subset of $S$ whose vertices 
are in the same partition as $s$, and 
$S_2 = S - S_1$.  
The weight of such a cut must be finite, as the maximum 
$s$-$t$ flow is bounded by $\min\; \{\; |S|, |C|\;\}$.  
Thus, we have that $N(S_1)$ and $N(S_2)$ are disjoint.   
Moreover, the total weight of the edges in the cut-set 
is $|S| - |S_1| + |N(S_1)|$, which implies that 
\[ S_1 = \arg\min_{S'}\; \{\; |N(S')| - |S'| \;\mid\; S' \subseteq S\; \}. \] 
So after running the flow algorithm to obtain $S_1$, 
there will be three cases:
\begin{description}
\item[Case 1:] the optimal $I^*$ for $a_0(G)$ is exactly $S_1$;
\item[Case 2:] the optimal $I^*$ for $a_0(G)$ is a proper subset of $S_1$;
\item[Case 3:] the optimal $I^*$ for $a_0(G)$ is a subset of $S_2$;
\end{description}
Note that Case 2 is impossible, 
since for any such proper subset $S_1'$, we have
\[  |N(S_1')| - |S_1'| \geq |N(S_1)| - |S_1|  
\tag{by min-cut} \]
which implies
\[  \frac{|N(S_1')| - |S_1'|}{|S_1'|} > \frac{|N(S_1')| - |S_1'|}{|S_1|} 
\geq \frac{|N(S_1)| - |S_1|}{|S_1|}, \]
so that
\[ \frac{|N(S_1')|}{|S_1'|} > \frac{|N(S_1)|}{|S_1|}. \]
Consequently, $S_1'$ cannot be an optimal set that achieves $a_0(G)$.

\bigskip 

Thus, we have either Case 1 or Case 3.  
To handle Case 3, we simply remove $S_1$ and $N(S_1)$ 
from the graph, and solve it recursively.  
In total, finding $a_0(G)$ requires $O(|S|)$ runs of the max-flow
algorithm, and can be solved in polynomial time.  

\begin{remark}
There exists a somewhat faster algorithm, 
also proposed by Cunningham~\cite{kn:cunningham}, 
which requires $O(\log |S|)$ runs of max-flow in a 
slightly different flow network; we omit the details for brevity.
\end{remark}

\medskip 

Finally, to compute $a_1(G)$, notice that,  
if an independent set $I$ contains some 
vertex $v \in C$ then $N(I)$ contains all vertices of $C$.  
When $|I|/(|I| + |N(I)|)$ is maximal,  
$I$ will contain all the vertices in $S$ that are 
nonadjacent to $v$. Hence  
\[ a_1(G) = \frac{n - d}{n}, \]
where $d$ denotes the minimum degree of a vertex in $C$. 
It follows that $a_1(G)$ can be obtained  
in linear time.

\bigskip 

This proves the following theorem. 
\begin{theorem}
There exists a polynomial-time algorithm to 
compute the tensor capacity for splitgraphs. 
\end{theorem}

\bigskip 

We modify the approach to obtain an 
exact algorithm for the tensor capacity of a general graph $H$.  
Let $n$ be the number of vertices in $H$.  
Assume we are given a maximal independent set $I$ of $H$. 
We let $I$ play the role of $S$ 
and $N(I)$ play the role of $C$ in the above transformation. 
Then, by the analysis above, we obtain a subset $I_1$ of $I$ with

\[ I_1 = \arg\max_{I'}\; 
\left\{\; \frac{|I'|}{|I'| + |N(I')|} 
\;\mid\; I' \subseteq I\; \right\}. \] 

\bigskip 

The algorithm generates 
all the maximal independent sets $I$s, 
and finds the corresponding subset $I_1$s for each of them.  
This yields the value $a(H)$. 
By Moon and Moser's classic result, 
$H$ contains at most $3^{n/3}$ maximal independent 
sets. Furthermore, by, eg, the algorithm of Tsukiyama et al., 
they can be generated 
in polynomial time per maximal independent set. 
Thus we obtain the following theorem. 

\begin{theorem}
There exists an $O^{\ast}(3^{n/3})$ algorithm to compute 
the tensor capacity for a graph with $n$ vertices. 
\end{theorem}

\bigskip 

\begin{remark}
We moved the section on the ultimate categorical 
independent domination ratio to Appendix~\ref{section domination ratio}. 
Appendix~\ref{hardness} contains the NP-completeness 
proof for $\alpha(G \times K_4)$ when $G$ is a planar graph 
of degree three. 
\end{remark}

\appendix

\section{The ultimate categorical independence ratio for 
some classes of graphs}
\label{section classes}

In this section we show that the tensor capacity is 
polynomial for permutation graphs, interval graphs, and graphs 
of bounded treewidth. The last result also shows that there 
is a PTAS for the ultimate categorical independence ratio 
of planar graphs. 

\subsection{The tensor capacity for permutation graphs}

A permutation diagram is obtained as follows. 
Consider two horizontal lines, $L_1$ and $L_2$, and 
label $n$ distinct points on each by $1,\dots,n$. For each 
label $i$ take the straight line segment that connects 
the points on $L_1$ and $L_2$ with that label. 
Pnueli et al. defined permutation graphs as follows~\cite{kn:pnueli}. 

\begin{definition}
A graph is a permutation graph if it is the intersection 
graph of the straight line segments in a permutation diagram. 
\end{definition}

Baker et al. characterized permutation graphs as follows~\cite{kn:baker}. 

\begin{theorem}
A graph $G$ is a permutation graph if and only if both $G$ and 
$\bar{G}$ are comparability graphs. 
\end{theorem}

\bigskip 

\begin{theorem}
There exists an $O(n^3)$ algorithm to compute the 
tensor capacity for permutation graphs. 
\end{theorem}
\begin{proof}
Consider a permutation diagram. Notice that an independent 
set consists of line segments that are parallel. 

\medskip 

\noindent 
For each line segment $x$, and for each integer $k$, 
compute the smallest neighborhood of an independent 
set of cardinality $k$ that has $x$ as its right-most 
line segment. 

\medskip 

\noindent
To compute this for $x$, consider the line segments $y$ 
that lie to the left of $x$. Let $N_{k-1}(y)$ 
be the smallest number of neighbors of an independent set 
with $k-1$ vertices that has $y$ as its right-most 
line segment. Let $N(x,y)$ be the number of 
neighbors of $x$ that are not neighbors of $y$. 
The value $N_k(x)$, for $k \in \mathbb{N}$, is defined as follows.  
\[N_k(x)= 
\begin{cases}
|N(x)| & \;\text{if $k=1$}\\
\min\;\{\;N_{k-1}(y)+N(x,y)\;|\; \text{$y$ lies to the 
left of $x$}\;\} & \text{otherwise.}
\end{cases}
\] 

\medskip 

\noindent 
The value 
\[a(G)=\max_{\text{$I$ an independent set}} \;\frac{|I|}{|I|+|N(I)|}\] 
is obtained by 
\[a(G)= \max\;\left\{\;\frac{k}{k+N_k(x)}\;\mid \; 
x\in V \quad k \in \mathbb{N}\;
\right\}.\]
The tensor capacity $\Theta^T(G)$ is obtained from 
Theorem~\vref{thm toth} via Formula~\eqref{eqn6} on 
Page~\pageref{eqn6}. 

\medskip 

\noindent
This proves the theorem. 
\qed\end{proof}

\subsection{The tensor capacity for interval graphs}

Haj\'os defined interval graphs as follows~\cite{kn:hajos}. 

\begin{definition}
An interval graph is an intersection graph of a collection of 
intervals on the real line.
\end{definition}

In the following we identify vertices and the intervals that represent 
them. 

\bigskip 

Notice that an independent set consists of a collection 
of disjoint intervals. So there is a linear, left-to-right 
ordering of the vertices of an independent set. 

\begin{definition}
Let $x$ be a vertex and let $k \in \mathbb{N}$. 
Let $I(x,k)$ denote the collection of 
independent sets of cardinality $k$ 
in which $x$ is 
the rightmost interval. 
Define 
\begin{equation}
i(x,k) = \min \; \{ \;|N(I)|\;|\;I \in I(x,k) \;\}.
\end{equation}
\end{definition}

\bigskip 
 
To compute $a(G)$, 
We can compute the value of $i(x,k)$ via the recurrence relation,
\begin{equation}
i(x,k) = \min_y \;\{\;i(y,k-1) + |N(x)\setminus N(y)|\;\} 
\end{equation}
where $y$ is one of the intervals whose right endpoint is 
to the left of the left endpoint of $x$. 
To avoid overcounting, we only add  
the neighbors of $x$ that are not neighbors of $y$. 
The correctness follows from the 
observation that if there is any interval overlapping 
with $x$ and another interval $z$ in the 
independent set,   
then $z$ must also overlap with $y$. 

\bigskip 

The algorithm computes $a(G)$ via the following formula.
\begin{equation}
a(G) = \max \; \left\{\;\frac{k}{k + i(x,k)}\;\mid\; x \in V \quad 
k \in \mathbb{N}\;\right\}.  
\end{equation}

\bigskip 

It is easy to see that the 
time complexity is bounded by $O(n^3)$. 
This proves the following theorem. 

\begin{theorem}
There exists an $O(n^3)$ algorithm to compute the 
tensor capacity for interval graphs. 
\end{theorem}

\begin{remark}
We leave it as an open problem whether the time complexity 
for interval graphs, or for permutation graphs, 
can be reduced to $O(n^2)$.  
\end{remark}

\subsection{The tensor capacity for graphs of bounded treewidth}

Graphs of bounded treewidth were popularized by 
Robertson and Seymour during their work on graph 
minors~\cite{kn:robertson}.  

\begin{definition}
A graph has treewidth at most $k$ if it is a subgraph of a chordal 
graph with clique number $k+1$. 
\end{definition}

For each $k$, the class of graphs of treewidth at most $k$ is closed 
under minors. The class plays a major role 
in the graph minor theory because every class of graphs that is 
closed under taking minors, which does not contain all 
planar graphs, has treewidth bounded by some $k$. 
The class of graphs with treewidth at most $k$ 
is recognizable in linear 
time~\cite{kn:bodlaender,kn:kloks3}. For some background information 
on this class of graphs we refer to, 
eg,~\cite{kn:bodlaender2,kn:kloks3}. 
 
\bigskip 

\begin{theorem}
Let $k \in \mathbb{N}$. 
There exists a polynomial-time algorithm that computes the tensor 
capacity for the class of graphs that have treewidth at most $k$. 
\end{theorem}
\begin{proof}
Consider a nice tree-decomposition of width $k$~\cite{kn:kloks3}. 
Each node of the decomposition tree is of four possible types. 
The algorithm computes a table which contains some  
information of the graphs induced by the vertices that appear 
in bags of the subtree. For these induced subgraphs the table contains, 
for each value $k$, the minimal number 
of neighbors of an independent set of cardinality $k$. 
each table entry further specifies 
\begin{enumerate}[\rm (a)]
\item the vertices in the bag 
that are contained in the independent set, and
\item  the vertices in the 
bag that are neighbors of vertices in the independent set. 
\end{enumerate}
We describe next how this information is computed for 
each type of the nodes in the tree-decomposition. 
\begin{description}
\item[Start node.] A start node $s$ is a leaf of the 
decomposition tree. In that case the induced subgraph is just the 
subgraph induced by the vertices that appear in the bag, say $S$. 
In that case, the table contains all the independent sets 
and all the neighbors of those independent sets. 
\item[Join node.]
A join node $s$ has exactly two children, say $s_1$ and $s_2$. 
The three bags are the same, say $S=S_1=S_2$. 
To construct the table for $S$, consider table entries 
at $s_1$ and $s_2$ that have identical independent set 
in $S_1$ and $S_2$. For the neighborhoods in $S$ the algorithm 
takes the union of the neighbors indicated in $S_1$ and $S_2$. 
The total number of vertices in the independent set is the sum 
of the numbers at the nodes $s_1$ and $s_2$, avoiding double 
counting the number that are in $S$. The number of neighbors is 
also the union of the neighbors in the subtree at $s_1$ and $s_2$, 
again avoiding double counting the neighbors that are in both 
$S_1$ and $S_2$. 
\item[Introduce node.]
An introduce node $s$ has exactly one child $s^{\prime}$. 
The bag $S$ of $s$ has exactly one vertex more than the bag $S^{\prime}$ 
of $s^{\prime}$. Say $S =S^{\prime} \cup \{x\}$. 
All neighbors of $x$ are in $S$. To compute the table at the node $s$ 
we consider the cases where $x$ is in the independent set, 
in the neighborhood of the independent set, or unrelated to the 
independent set. Since all neighbors of $x$ are in $S$, the 
table entries at $s^{\prime}$ are easily extended to make up 
table entries for the node $s$. 
\item[Forget node.]
A node $s$ is a forget node if it has exactly one child, 
say $s^{\prime}$, and the bag of $S$ has exactly one vertex $x$ 
less than the bag $S^{\prime}$. Say $S^{\prime} =S \cup \{x\}$. 
The table at $s$ is easily obtained from the table at $s^{\prime}$. 
The values for the independent sets and their numbers of neighbors 
don't change. Simply the information whether the vertex $x$ is a 
vertex of the independent set, or if it is a neighbor 
of the independent set, or if it is unrelated 
to the independent set, disappears. Of course, 
this may cause some table entries 
to coincide.        
\end{description}
This describes the dynamic programming algorithm. 
The timebound is determined by the size of the tables. 
each table entry is characterized by a 3-coloring of the 
vertices in the bag; namely as a vertex of the independent set, 
as a neighbor of the independent set, or as a vertex which is not related 
to the independent set. Since each bag contains at most $k+1$ vertices, 
there are $O(3^{k+1})$ different types. For each type, the table 
entry contains two numbers, namely the total size of the 
independent set and the total number of neighbors. 
Thus the size of each table is bounded by $O(3^{k+1}\cdot n^2)$. 

\medskip 

\noindent
The decomposition tree has $O(n)$ different nodes. 
Each table is computed in constant time per table entry. 
Thus the total time is bounded by $O(3^{k+1}\cdot n^3)$ time. 
\qed\end{proof}

\bigskip 

Via Baker's method we easily obtain the following 
result~\cite{kn:baker2}. For brevity we omit the (standard) details.  

\begin{theorem}
There exists a PTAS to approximate the ultimate 
categorical independence ratio in planar graphs. 
\end{theorem}

\section{NP-Completeness of independence in categorical 
products of planar graphs}
\label{hardness}

\begin{theorem} 
\label{thm:npc}
Let $G$ be a planar graph of maximum vertex degree 3. 
It is NP-complete to 
compute the maximum independent set of $G \times K_4$.
\end{theorem}

\begin{proof}
Clearly, the problem is in NP.
We show that, to decide whether 
there is an independent set of size $4k$ is NP-hard for 
for $G \times K_4$ when $G$ is a planar graph of 
maximal degree three.

\medskip 

\noindent
We reduce the decision problem of deciding whether 
there is an independent set for $G$ of size $k$, 
which is known to be NP-complete~\cite{planar-npc}, to this problem.

\medskip 

\noindent
Let $K_4 = \{ t_1, t_2, t_3, t_4 \}$.
Now suppose that if $G$ has an independent set $S$ of size $k$,
then for each vertex $s$ in $S$, we select the four vertices 
$(s,t_1), (s,t_2), (s,t_3),$ and $(s,t_4)$ in $G \times K_4$.
Clearly, the selected $4s$ vertices $S'$ is an 
independent set in $G \times K_4$.

\medskip 

\noindent
On the other hand, suppose that $G \times K_4$ has an 
independent set $S'$ 
of size $4s$ in $G \times K_4$. 
Unfortunately, the related vertices $S$ in $G$ corresponding
to the vertices in $S'$ are not necessarily independent. 
We transform $S$ 
so that it becomes independent. 

\medskip 

\noindent 
For any vertex $s$ in $S$, 
it has at most three neighbors in $S$,
say $s_1, s_2$ and $s_3$. 
Without loss of generality, 
we assume that $(s,t_4)$ belongs to $S'$.
Then clearly, any of $(s_i, t_j)$ where $i=1,2,3$ and $j=1,2,3$ 
does not belong to $S'$
because otherwise, $(s_i, t_j)$ would be 
adjacent to $(s, t_4)$, which is impossible
since $S'$ is an independent set.
Thus the three vertices $(s_1, t_4)$, $(s_2, t_4)$ and $(s_3, t_4)$ 
must all belong to $S'$.

\medskip 

\noindent
We transform these three vertices in $S'$ to 
become $(s, t_1)$, $(s, t_2)$ and $(s, t_3)$.
It is clear that the resulted $S'$ is still independent.
Consequently, we also remove all $s_1, s_2, s_3$ from $S$.
If the new $S$ is not an independent set, 
then we apply the above transformation step on
another vertex in $S$ which has at least one more neighbor in $S$.
At the end, we obtain an independent set $S$ such that $|S| \geq 4k/4 = k$
since the size of $K_4$ is 4.

\medskip

\noindent
This completes our hardness proof.
\qed\end{proof}

\section{The ultimate categorical independent domination ratio 
for complete multipartite graphs}
\label{section domination ratio}

In this section we assume that the graphs have no isolated 
vertices. 

\begin{definition}
Let $G$ be a graph. The independent domination number 
$i(G)$ is the smallest cardinality of an independent dominating set 
in $G$. That is, $i(G)$ is the cardinality of a smallest 
maximal independent set in $G$. 
\end{definition}

\bigskip 

In~\cite{kn:farber}, 
Farber studies the following `independent domination capacity' 
for the strong product $G \boxtimes \dots \boxtimes G$.  
\[i_s(G)=\lim_{k \rightarrow \infty} \sqrt[k]{i(\boxtimes^k G}).\]  
For chordal graphs $G$ the fractional independent domination number,  
$i_f(G)$,   
equals the independent domination number. Farber shows that there 
are infinitely many trees $T$ for which $i_s(T) < i(T)$. 
It seems difficult to get a grip on the 
parameter. Farber conjectures that $i_s(C_4) = \sqrt[3]{4}$. 

\bigskip 

In the rest of this section we concentrate on the 
categorical product. 
To start with,    
the following conjecture appears in~\cite{kn:nowakowski}. 

\begin{conjecture}
\label{conj 2}
For all graphs $G$ and $H$ 
\begin{equation}
\label{eqn12}
i(G \times H) \geq i(G) \cdot i(H).
\end{equation}
\end{conjecture}

\bigskip 

\begin{definition}
The independent domination ratio of a graph $G$ 
is defined as 
\begin{equation}
\label{eqn13}
r_i(G) = \frac{i(G)}{|V(G)|}.
\end{equation}
\end{definition}

In~\cite[Section 5.2.1]{kn:finbow}, 
Finbow studies the ultimate categorical independent domination 
ratio. 

\begin{lemma}
\label{lm indep dom prod}
Let $G$ and $H$ be graphs without isolated vertices. 
Then 
\begin{equation}
\label{eqn14}
i(G \times H) \leq i(G) \cdot |V(H)|.
\end{equation}
\end{lemma}
\begin{proof}
Let $A$ be a minimum independent dominating set in $G$. 
Let 
\[S=\{\;(a,h)\;|\; a \in A \quad\text{and}\quad h \in V(H)\;\}.\] 
Since $H$ has no isolated vertices, 
$S$ is a dominating set in $G \times H$  
with cardinality 
\[|S|=i(G)|V(H)|.\]  
This proves the lemma.
\qed\end{proof}

\begin{lemma}
\label{indep dom cap}
The sequence $r_i(G^k)$, $k \in \mathbb{N}$, is 
non-increasing. Thus the limit 
\[I(G)=\lim_{k \rightarrow \infty} r_i(G^k)\] 
exists.
\end{lemma}
\begin{proof}
Notice that, by Lemma~\ref{lm indep dom prod},  
\[i(G^k) \leq i(G^{k-1}) \cdot |V(G)| 
\Rightarrow 
\frac{i(G^k)}{|V(G)|} \leq i(G^{k-1}) 
\Rightarrow 
\frac{i(G^k)}{|V(G)|^k} \leq \frac{i(G^{k-1})}{|V(G)|^{k-1}}.\]
This proves the claim. 
\qed\end{proof}

\bigskip 

\begin{remark}
Finbow shows in~\cite[Page 57]{kn:finbow} 
that for the complete bipartite graph 
\[I(K_{m,m})= \lim_{k \rightarrow \infty} r_i(K^k_{m,m})= 
\lim_{k \rightarrow \infty} \frac{i(\times^k K_{m,m})}{(2m)^k} = 
\lim_{k \rightarrow \infty} \frac{2^{k-1}\cdot m^k}{2^k \cdot m^k} =
\frac{1}{2}.\]
\end{remark}

\bigskip 

\begin{lemma}
\label{lm unbalanced cb}
Let $G \simeq K(m,n)$ be the complete bipartite graph with 
$m$ and $n$ vertices in the two color classes. 
Then 
\begin{equation}
\label{eqn15}
r_i(G^k)= \frac{1}{(m+n)^k} \cdot \sum_{\ell=0}^{k-1} 
\binom{k-1}{\ell} \cdot \min\;\{\;m^{k-\ell}n^{\ell},\;m^{\ell} n^{k-\ell}\;\}.
\end{equation}
This implies that $I(K_{m,n})=0$ 
when $0 < m < n$. 
\end{lemma}
\begin{proof}
According to~\cite[Page~57]{kn:finbow}, 
\[K(m,n) \times K(p,q)=K(mp,nq) \oplus K(mq,np).\]
Via induction, it follows that 
\[\times^k K(m,n)=\sum_{\ell=0}^{k-1} \binom{k-1}{\ell} \;  
K(m^{k-\ell}n{^\ell},m^{\ell}n^{k-\ell}),\] 
where the sum denotes union. 
For a complete bipartite graph $K(p,q)$ the 
independent domination number is $\min\{p,q\}$. 

\medskip 

\noindent 
Let $m=\alpha \cdot n$ for some $0 < \alpha < 1$. 
Then~\eqref{eqn15} yields 
\begin{eqnarray*}
I(G) &=& \frac{1}{(m+n)^k} \cdot \left[ 
\sum_{0\leq \ell \leq k/2} \binom{k-1}{\ell} m^{k-\ell} n^{\ell} 
+ \sum_{1 \leq \ell < k/2} \binom{k-1}{\ell-1} m^{k-\ell}n^{\ell} \right]  \\
&\leq & \frac{1}{(m+n)^k} \cdot 
\sum_{0 \leq \ell \leq k/2} \binom{k}{\ell} m^{k-\ell}n^{\ell}  \\
&\leq& \sqrt{k/(2\cdot \pi)}\cdot \left( \frac{2 \sqrt{\alpha}}{1+\alpha} 
\right)^k \rightarrow 0 \quad (k \rightarrow \infty).
\end{eqnarray*}
This proves the lemma.   
\qed\end{proof}

\begin{theorem}
Let $G$ be a complete multipartite graph with $t$ color classes 
of size 
\[n_1 \leq \dots \leq n_t.\] 
Then $I(G)=0$ unless $t=2$ and $n_1=n_2$, 
in which case $I(G)=\frac{1}{2}$. 
\end{theorem}
\begin{proof}
For the case where $G \simeq K(m,m)$, Finbow proved that 
$I(G)=\frac{1}{2}$~\cite{kn:finbow}. 
When $t=2$ and $n_1 < n_2$ then $I(G)=0$, as is shown 
in Lemma~\ref{lm unbalanced cb}. 

\medskip 

\noindent
Assume $t \geq 3$. Let $G^{\prime}$ be the subgraph of 
$G$ obtained from $G$ by removing all edges 
except those with one endpoint in the smallest 
color class. Then, obviously,   
\[i(G^{\prime}) \geq i(G) \quad\text{and}\quad 
i((G^{\prime})^k) \geq i(G^k).\]
The graph $G^{\prime}$ is complete bipartite and the two 
color classes do not have the same size. Therefore, 
\[\lim_{k \rightarrow \infty} r_i((G^{\prime})^k) = 0 
\quad\Rightarrow\quad 
I(G) = \lim_{k \rightarrow \infty} r_i(G^k) =0.\]
This proves the theorem.
\qed\end{proof}

\end{document}